\documentclass[acmtog,authorversion,nonacm]{acmart}
\acmJournal{TOG}
\acmSubmissionID{}
\acmPrice{}
\acmDOI{}
\setcitestyle{square}
\setcopyright{none}
\authorsaddresses{}
\usepackage{comment}
\usepackage[utf8]{inputenc}
\usepackage{amsmath}
\usepackage{amsfonts}
\usepackage{amsthm}
\usepackage{bm}
\usepackage{color}
\usepackage{nicefrac}
\usepackage{dsfont}
\usepackage[section]{placeins}
\usepackage{overpic}
\usepackage{wrapfig}

\definecolor{commentgray}{gray}{0.5}
\usepackage[ruled,vlined,noend]{algorithm2e}
\usepackage{setspace}

\SetCommentSty{mycommfont}
\SetAlCapHSkip{1.1mm}

\usepackage{graphicx}
\usepackage{subfigure}
\usepackage{csquotes}
\usepackage{cleveref}
\crefname{section}{Section}{Sections}
\crefname{subsection}{Section}{Sections}
\crefname{figure}{Figure}{Figures}
\crefname{prop}{Prop.}{Props.}
\crefname{equation}{Eq.}{Eqs.}
\crefname{algorithm}{Algorithm}{Algorithms}
\crefname{table}{Table}{Tables}
\crefname{defn}{Def.}{Defs.}

\newcommand*{\ve}[1]{\bm{#1}}

\newcommand{\cT}{M}
\newcommand{\talpha}{\alpha}
\newcommand{\tl}{\ell'}
\newcommand{\RV}{V}

\newtheorem{defn}{Definition}
\newtheorem{prop}{Proposition}

\definecolor{plotblue}{HTML}{5a7daf}
\definecolor{plotgreengreen}{HTML}{8ed659}
\definecolor{plotgreen}{HTML}{8bc284}
\definecolor{plotbluegreen}{HTML}{88aeaf}
\definecolor{plotred}{HTML}{f08585}
\definecolor{backgray}{HTML}{f0f0f0}

\definecolor{plot1}{HTML}{5a7daf}
\definecolor{plot2}{HTML}{6596c6}
\definecolor{plot3}{HTML}{69b2ba}
\definecolor{plot4}{HTML}{7dc5b6}
\definecolor{plot5}{HTML}{c4e1ae}

\author{Marcel Campen}
\affiliation{%
  \institution{Osnabrück University}
  \country{Germany}}
\author{Ryan Capouellez}
\author{Hanxiao Shen}
\author{Leyi Zhu}
\author{Daniele Panozzo}
\author{Denis Zorin}
\affiliation{%
  \institution{NYU}
  \country{USA}}

\title{Efficient and Robust Discrete Conformal Equivalence with Boundary}

\begin{document}

\begin{abstract}
We describe an efficient algorithm to compute a conformally equivalent metric for a discrete surface, possibly with boundary, exhibiting prescribed Gaussian curvature at all interior vertices and prescribed geodesic curvature along the boundary. Our construction is based on the theory developed in  \cite{gu2018discrete2,gu2018discrete,springborn2019ideal}, and in particular relies on results on hyperbolic Delaunay triangulations. 
Generality is achieved by considering the surface's intrinsic triangulation as a degree of freedom, and particular attention is paid to the proper treatment of surface boundaries. 
While via a double cover approach the boundary case can be reduced to the closed case quite naturally, the implied symmetry of the setting causes additional challenges related to stable Delaunay-critical configurations that we address explicitly in this work.
\end{abstract}

\maketitle

\section{Introduction}

Let $\cT$ be a triangle mesh, equipped with a discrete Euclidean metric, i.e., an assignment of length to edges, satisfying triangle inequality, forming a surface with or without boundary. 
Under this metric, let $\alpha_{jk}^i$ be the angle at vertex $v_i$ in the 
triangle $T_{ijk}$. Let $\Theta_i = \sum_{T_{ijk}}
\alpha^i_{jk}$ 
be the total angle at vertex $v_i$.
Define $\kappa_i$ as the angle deficit at a vertex $v_i$, 
defined as $2\pi - \Theta_i$ for interior vertices  and $\pi - \Theta_i$ for boundary vertices.
This quantity can be viewed as the discrete Gaussian curvature if $v_i$ 
is an interior vertex and the geodesic curvature of the boundary if $v_i$ is on the 
boundary. 
 
Given \emph{target} curvatures $\hat\kappa_i$ (respecting the discrete Gauss-Bonnet theorem) one may ask for a discrete metric that exhibits exactly these curvatures. This is of practical interest, for instance, to obtain flattenings, i.e., surface parametrizations over the plane (by prescribing $\kappa_i = 0$ in the interior \cite{BenChen:2008}) or so-called seamless maps for quadrilateral remeshing (by prescribing $\hat\kappa_i = k_i\frac{\pi}{2}$ with $k_i \in \mathbb{Z}$ \cite{campen2018seamless,myles2012global}).
Such a metric always exists, and when restricting to metrics conformally equivalent to the original metric, it is unique (up to scale)
\footnote{For surfaces of non-trivial topology, the more general prescription of the metric's associated \emph{holonomy} or \emph{monodromy} is of interest. For this case similar statements hold for metrics with scale jumps (so-called \emph{similarity structures}) \cite{Campen:2017:SimilarityMaps}.}.

The conformal case is of particular relevance because, in principle, such a metric can be found via solving a convex optimization problem \cite{springborn2008conformal}. For a fixed triangulation, however, the triangle inequality limits how much $\hat\kappa$ may differ from $\kappa$ before the problem becomes infeasible \cite{springborn2008conformal}. By treating the surface's triangulation as variable, the problem can be made feasible in general,  as shown in  \cite{gu2018discrete,gu2018discrete2,springborn2019ideal}. 
The vertex set $V$ can be kept fixed, i.e., no refinement is necessary and intrinsic edge flips are sufficient to facilitate all required adjustments. A bijection between the original and the modified triangulation can easily be kept track of \cite{fisher2007algorithm}, so as to, e.g, in the end extend the computed conformal parametrization from the modified triangulation to (a~refinement~of) the original mesh.

Recent results \cite{gu2018discrete,gu2018discrete2,springborn2019ideal} indicate how these triangulation changes can be performed in a systematic manner. We discuss the relevant background (\cref{sec:background}) and describe an implementation of these ideas, with particular attention to practical aspects (\cref{sec:algo}), as well as a generalization to surfaces with boundary, which poses remarkable additional challenges (\cref{sec:boundary}).

\section{Related Work}

The problem of computing conformally equivalent metrics or, by implication, conformal maps of discrete surfaces has been considered in a variety of works before. As there is no useful natural notion of conformality in the discrete (non-smooth) setting, a range of discrete counterparts of the continuous concept of conformality have been proposed and used.

\paragraph{Static Triangulation}
Prominent examples of works addressing the computation of conformal metrics or maps, based on various definitions of discrete conformality, on discrete surfaces while considering their triangulation fixed are: \cite{springborn2008conformal,Gu:2003,BenChen:2008,BFF,levy2002least,soliman2018optimal,kharevych2006discrete,jin2007discrete,desbrun2002intrinsic}.

\paragraph{Dynamic Triangulation}
A fixed triangulation restricts the space of metrics that can be achieved. For instance, a vertex $v_i$ of valence $k$ cannot, under any (Euclidean) metric, exhibit a discrete curvature $\kappa_i \leq (2-k)\pi$, as inner angles are bounded by $\pi$. By adjusting the triangulation depending on the prescribed target curvature, this limitation can be remedied. Two systematic approaches have been proposed to that end, both conceptually considering a continuous metric evolution from initial state to target state. \cite{luo2004combinatorial} proposes to adjust the triangulation by an intrinsic edge flip whenever an edge becomes triangle inequality critical (\cref{fig:flips} left). Implementation variants are described in \cite{Campen:2017:SimilarityMaps,Campen:2017:OnSimilarityMaps,campen2018seamless}. Differently,
\cite{gu2018discrete,gu2018discrete2,springborn2019ideal} effectively consider the case of flipping an edge when it becomes Delaunay-critical (\cref{fig:flips} right).

\section{Background}
\label{sec:conformal}
\label{sec:background}

We begin by considering the case of surfaces \emph{without} boundary, i.e., we are given a closed manifold triangle mesh $\cT = (V,E,F)$; this is the setting considered 
in \cite{gu2018discrete} and other related work. 
The case of surfaces with boundary can be reduced to the closed surface case with an additional symmetry structure, as we show in \cref{sec:boundary}.

The mesh $\cT$ is equipped with
an input metric defined by
edge lengths $\ve{\ell} : E \rightarrow \mathbb{R}^{>0}$, 
satisfying the triangle inequality. 

\subsection{Conformal Equivalence}
A \emph{conformally equivalent} discrete metric is defined by means of \emph{logarithmic scale factors} $\ve{u} : V \rightarrow \mathbb{R}$ associated with vertices $V = (v_1,\dots,v_n)$, by defining new edge lengths as
\begin{equation}
    \ell_{ij}(\ve{u}) = \ell_{ij} \,e^{\frac{u_i+u_j}{2}} 
    \label{eq:barell}
\end{equation}
per edge $e_{ij}$ \cite{luo2004combinatorial}.
Given per-vertex
target angles $\hat\Theta_i$
a conformally equivalent metric exhibiting these is characterized by, for all~$i$:
\begin{equation}
    g_i(\ve{u}) := \hat{\Theta}_i - \Theta_i(\ve{u}) = \hat{\Theta}_i -  \sum_{T_{ijk}} \talpha^i_{jk}(\ve{u}) = 0,
    \label{eq:flattening1}
\end{equation}
where the inner angle $\talpha^i_{jk}(\ve{u})$ is computed under the metric defined by $\ve{u}$ via \cref{eq:barell} (i.e. from edge lengths $\ve{\ell(u)}$).

It is known that $\ve{g}(\ve{u}) = (g_1(\ve{u}), \dots, g_n(\ve{u}))$ is the gradient of a twice-differentiable convex function \cite{springborn2008conformal}. Hence, one may constructively yield factors $\ve{u}$ satisfying \cref{eq:flattening1} using (second-order) convex optimization methods, starting from arbitrary initializations (e.g. $\ve{u}\equiv \ve{0}$). This is true, however, only as long as $\ve{u}$ stays in the feasible region $\Omega \subset \mathbb{R}^n$ where $\ve{\ell(u)}$ respects the triangle inequality for each triangle~$T_{ijk}$; otherwise it does not well-define a Euclidean surface metric on $\cT$.

\subsection{Dynamic Triangulation}
\label{sec:dyn}

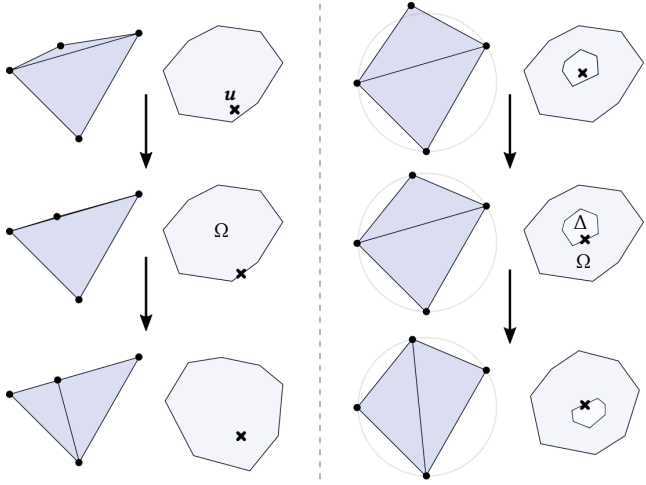
\begin{figure}[t]
    \centering
      \begin{overpic}[width=.99\linewidth]{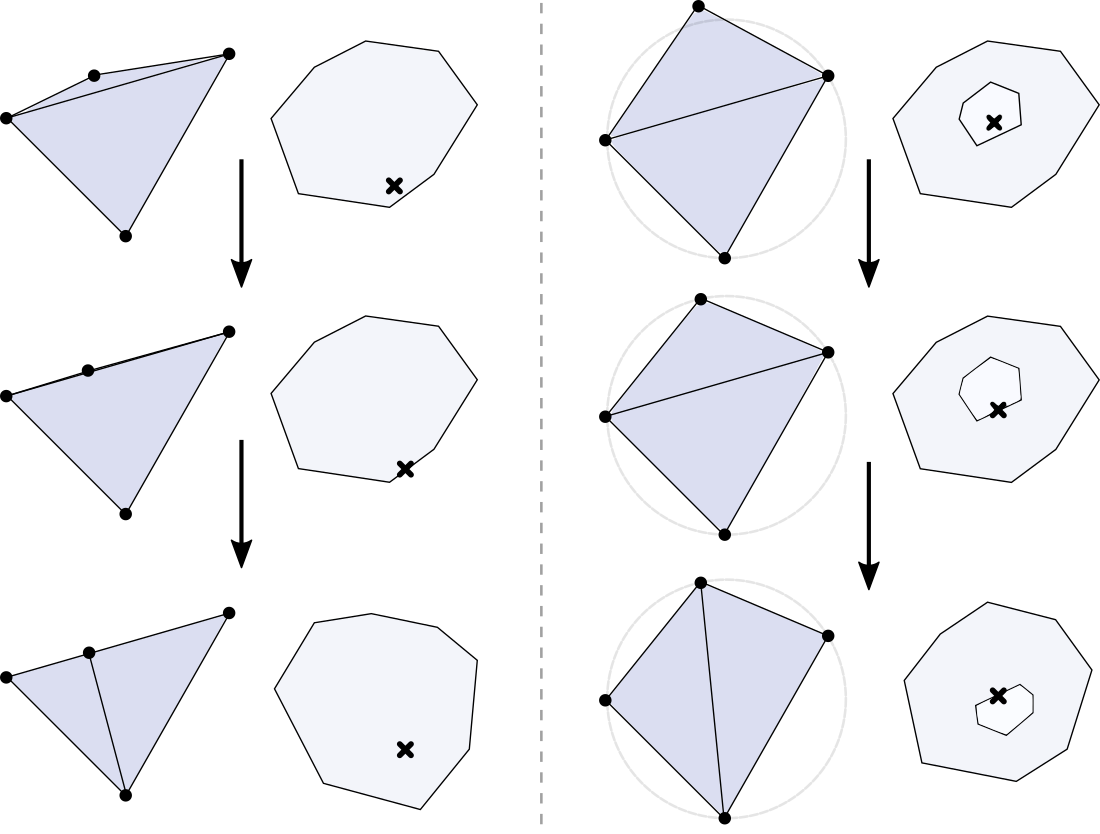}
    \small
    \put(32.5,39){$\Omega$}
    \put(89.4,34){$\Omega$}
    \put(89.1,40.2){$\Delta$}
    \put(34,60.5){$\ve{u}$}
    \end{overpic}
    \caption{Left: flip-on-degeneration. Right: flip-on-Delaunay-violation. Alongside a conceptual illustration of the valid region $\Omega$ (light blue) and Delaunay region $\Delta$ (white) is shown (cf.~\cref{sec:dyn}), containing the current point $\ve{u}$ (cross mark) and changing due to the flip.}
    \label{fig:flips}
\end{figure}

The feasible region $\Omega$ can be altered by choosing a different triangulation of the same surface. There may not be a common triangulation, though, on which both $\ve{\ell} \equiv \ve{\ell}(\ve{0})$ and $\ve{\ell}(\ve{u}^*)$ for the sought $\ve{u}^*$ (satisfying \cref{eq:flattening1}) are valid. Rather, the triangulation needs to be adjusted dynamically \emph{during} the evolution of $\ve{u}$ from $\ve{0}$ towards $\ve{u}^*$.

Note that a change of triangulation is possible without \emph{intrinsically} changing the surface. $\cT$ together with given edge lengths defines a surface $S_V$ with a metric which is flat everywhere except at $V$. There are many triangulations (besides $\cT$) with vertices $V$ and their own associated edge lengths, defining the same surface $S_V$ (cf.~\cite{sharp2019navigating}); hence the differentiation between $\cT$ and $S_V$. In particular, an edge flip replacing a pair of triangles  $(T_{ijk}, T_{jim})$ sharing an edge $e_{ij}$, with triangles $(T_{kim}, T_{mjk})$ sharing edge $e_{km}$ can be performed without intrinsically changing the surface $S_V$, by setting the length of the new edge $e_{km}$ to the length of the diagonal of the planar quadrilateral obtained by unfolding $T_{ijk}, T_{jim}$ \cite{fisher2007algorithm}. This is referred to as \emph{intrinsic flip}.

\paragraph{Delaunay Flips}
\label{sec:delflips}
 \cite{gu2018discrete,springborn2019ideal} propose to dynamically adjust the triangulation such that it is (intrinsically) Delaunay at all times as $\ve{u}$ evolves.

\begin{defn}[Intrinsic Delaunay]
A triangulation is intrinsic Delaunay
if any two triangles $T_{ijk}$ and $T_{jim}$
 sharing an edge $e_{ij}$ satisfy the Delaunay condition:
 \begin{equation}
    \cos\alpha'^k_{ij} + \cos\alpha'^m_{ij} \geq 0
    \label{eq:delaunay}
\end{equation}
where $\alpha'^k_{ij}$ and $\alpha'^\ell_{ij}$ are triangle angles opposite edge $e_{ij}$.
Expressed directly in terms of edge lengths this condition is equivalent to
 \begin{equation}
 \frac{\ell'^2_{jk}+\ell'^2_{ki}-\ell'^2_{ij}}{\ell'_{jk}\ell'_{ki}} + \frac{\ell'^2_{jm}+\ell'^2_{mi}-\ell'^2_{ij}}{\ell'_{jm}\ell'_{mi}} \geq 0.
 \label{eq:delaunay_ell}
 \end{equation}
\end{defn}
In the context at hand these angles or lengths are to be understood as dependent on $\ve{u}$; we use short-hands $\alpha' = \alpha(\ve{u})$ and $\ell' = {\ell}(\ve{u})$.

Generically (iff these weak inequalities hold strictly), the intrinsic Delaunay triangulation is unique, but for special configurations (four or more intrinsically co-circular vertices resulting in equality in \cref{eq:delaunay_ell}) it is not.

For a given triangulation, let $\Delta \subset \mathbb{R}^n$ (referred to as \emph{Penner cell}) denote the region of factors $\ve{u}$ such that the triangulation is intrinsic Delaunay. Clearly, $\Delta \subset \Omega$, and when $\ve{u} \in \partial\Delta$ the Delaunay triangulation is not unique. Whenever $\ve{u}$ reaches the boundary of $\Delta$, we can switch to another Delaunay triangulation by means of an intrinsic flip, thereby changing the region $\Delta$ (and $\Omega$), enabling $\ve{u}$ to evolve further without leaving $\Delta$. \cref{fig:flips} right illustrates this behavior. Remarkably, these cells form a partition of $\mathbb{R}^n$. 

This can be formalized by the following definition of discrete conformal equivalence of two metrics \cite{gu2018discrete}: 

\begin{defn}[Discrete Conformal Equivalence]
Two metrics $(\cT_1, \ell_1)$ and $(\cT_m,\ell_m)$ are \emph{discretely conformally equivalent}, if 
there is a sequence of meshes with the same vertex set,  $(\cT_s,\ell_s$), $s=1,\ldots,m$, 
such that, for each $s$, $\cT_s$ is an intrinsic Delaunay triangulation for the metric $\ell_s$ 
and either 
\begin{itemize}
\item $(\cT_s,\ell_s)$ and $(\cT_{s+1},\ell_{s+1})$ are different metrics with the same triangulation (i.e., $\cT_s$ = $\cT_{s+1}$) and the edge lengths are related by \cref{eq:barell}
for a choice of $u_s: V \rightarrow \mathbb{R}$.
\item $(\cT_s,\ell_s)$ and $(\cT_{s+1},\ell_{s+1})$ are different Delaunay triangulations for the same metric.
\end{itemize}
\end{defn}

\paragraph{Degeneration Flips}
\label{sec:degenflips}
The alternative of 
performing a triangulation change only when $\ve{u}$ reaches the boundary $\partial\Omega$ of the currently feasible region was considered by \cite{luo2004combinatorial}.
This occurs when a triangle becomes a degenerate cap. An intrinsic flip of this triangle's longest edge yields a non-degenerate triangulation, effectively changing the valid region $\Omega$ such that $\ve{u}$ lies strictly in its interior. \cref{fig:flips} left illustrates this. An implementation of this approach is described and applied in \cite{Campen:2017:SimilarityMaps}.\medskip

At first sight, the approach based on maintaining an intrinsic Delaunay triangulation may seem inefficient in comparison. Due to $\Delta \subset \Omega$, at least as many, but often many more cells $\Delta$ need to be traversed.
Practically, this suggests a large number of small steps between flips in the process of optimizing $\ve{u}$, compared to, e.g., the use of (less frequent) degeneration flips, and much smaller steps compared to typical unconstrained optimization. 
Remarkably, however, this Delaunay approach permits an implementation that is in general more efficient and more robust (see \cref{sec:res:comparison} for a comparison).
In essence, exploiting a relation to \emph{hyperbolic} Delaunay triangulation, arbitrarily large steps can be made, beyond $\Delta$ and even beyond $\Omega$ (unconstrained by Euclidean triangle inequalities). Flips can be performed collectively \emph{after the fact} and in \emph{arbitrary order}. This is detailed in \cref{sec:hyp}.

\subsection{Evolution Step}
Assume we are given a triangulation $M$ that is intrinsic Delaunay under the metric defined by some $\ve{u}_\vdash$.
Consider a linear evolution of $\ve{u}$ from $\ve{u}_\vdash$ to $\ve{u}_\dashv$:
$$\ve{u}(t) = (1-t)\ve{u}_\vdash +t \ve{u}_\dashv,\; t\in[0,1].$$
As we move along the interval $[0,1]$, whenever four vertices
forming triangles $T_{ijk}$ and $T_{jim}$ become co-circular 
under the metric defined by $\ve{\ell}(\ve{u}(t))$, an intrinsic flip of edge $e_{ij}$ is performed. Due to the special configuration (the two triangles forming an \emph{inscribed} quadrilateral, see \cref{fig:ptolemy}) the length that the new edge $e_{km}$ needs to take can be computed following Ptolemy's theorem as 
\begin{equation}
\tl_{km}= \frac{1}{\tl_{ij}}(\tl_{jk} \tl_{im} +  \tl_{ki}\tl_{mj}),
\label{eq:ptolemy}
\end{equation}
where we use $\ve{\tl}$ as a shorthand for $\ve{\ell}(\ve{u}(t))$. For $\tl_{km} = \ell_{km}(u(t)) = \ell_{km}\,e^\frac{u_k+u_m}{2}$ to take this value for the current $\ve{u}(t)$, we need to set $\ell_{km}$ accordingly. If we plug \cref{eq:barell} into \cref{eq:ptolemy}, 
\begin{equation}
    \tl_{km}= \frac{1}{\ell_{ij}}(\ell_{jk} \ell_{im}+ \ell_{ki}\ell_{mj}) e^{(u_k + u_m)/2},
\notag
\end{equation}
we see that we need to set
\begin{equation}
\ell_{km} := \frac{1}{\ell_{ij}}(\ell_{jk} \ell_{im}+ \ell_{ki}\ell_{mj}).
\label{eq:ptolemy-orig}
\end{equation}
Notice that this is Ptolemy's formula, \cref{eq:ptolemy}, applied to the \emph{original} metric. In other words: applying the formula in the current ($\ve{u}(t)$-scaled) metric $\ve{\tl}$ is equivalent to applying it in the original metric $\ve{\ell}$, followed by scaling. Remarkably, this holds \emph{even though the vertices are not co-circular under the original metric} in general. Moreover, the edge lengths $\ve{\ell}$ set in this way may \emph{not even satisfy the triangle inequality}. This is no issue, though, as certainly the relevant scaled lengths $\ve{\tl}=\ve{\ell}(\ve{u}(t))$ do, by construction.

It was shown that the number of flip events along the path is finite \cite{WuFinite}, which means that after a finite number of flips we will 
obtain the triangulation and edge length assignment needed for the target $\ve{u}(1) = \ve{u}_\dashv$.

One practical downside of this procedure, in which the necessary flips along the evolution path are detected and performed one-by-one
sequentially, is that it requires solving precisely for the sequence of flips. This makes it inefficient as well as potentially numerically challenging.
The following \emph{hyperbolic} approach, whose correctness can be shown based on an interpretation of the involved edge lengths as defining hyperbolic metrics instead of Euclidean metrics, improves on this.

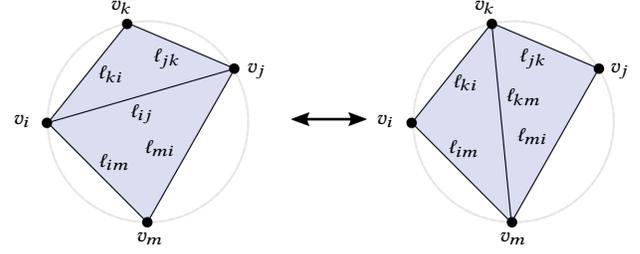
\begin{figure}[tb]
    \centering
        \vspace{0.3cm}
      \begin{overpic}[width=.89\linewidth]{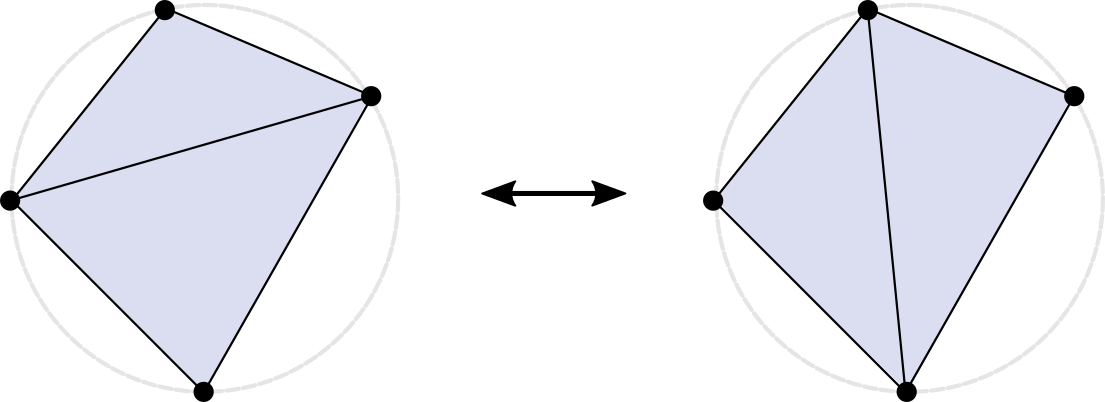}
    \small
    \put(15.5,19.3){$\ell_{ij}$}
    \put(81.2,22.0){$\ell_{km}$}
    
    \put(19.5,29.0){$\ell_{jk}$}
    \put(9.9,26.0){$\ell_{ki}$}
    \put(10.0,10.5){$\ell_{im}$}
    \put(18,13.0){$\ell_{mi}$}
    
    \put(83.5,29.0){$\ell_{jk}$}
    \put(71.9,24.5){$\ell_{ki}$}
    \put(71.1,12.9){$\ell_{im}$}
    \put(83.1,15.0){$\ell_{mi}$}
    
    \put(-5,18.0){$v_i$}
    \put(35.8,27.1){$v_j$}
    \put(12,38){$v_k$}
    \put(16.5,-2.5){$v_m$}
    
    \put(58.5,18.0){$v_i$}
    \put(99.3,27.1){$v_j$}
    \put(75.5,38){$v_k$}
    \put(80,-2.5){$v_m$}
    \end{overpic}
    \vspace{0.2cm}
    \caption{Ptolemy flip of an edge $e_{ij}$ shared by two triangles forming an inscribed quadrilateral, i.e., a Delaunay-critical edge.}
    \label{fig:ptolemy}
\end{figure}

\subsection{Hyperbolic Metric Approach}
\label{sec:hyp}

Instead of moving $t$ along the interval $[0,1]$, determining the sequence of flip events and executing them in order, let us directly consider $t=1$.
The initial triangulation $M$ may not be Delaunay under $\ve{u}(1)$, and the edge lengths $\ve{\ell}(\ve{u}(1))$ may not even respect the triangle inequality.
Nevertheless, we can test each edge for violation of the Delaunay criterion using \cref{eq:delaunay_ell} applied to $\ve{\ell}(\ve{u}(1))$, and incrementally flip (using \cref{eq:ptolemy-orig}) all violating edges in arbitrary order following the classical flip algorithm until a Delaunay triangulation is reached \cite{bobenko2007discrete}. While in case of triangle inequality violations this criterion lacks the geometric justification via \cref{eq:delaunay} (the involved quantities are no longer cotangents of Euclidean angles), this algorithm succeeds anyway; this is based on a remarkable relation between conformal metric equivalence and hyperbolic isometries 
used in the constructions of \cite{gu2018discrete} and \cite{springborn2019ideal}. 

In summary, instead of performing flips following an expensive-to-compute sequence required to maintain a valid Euclidean metric on triangles at all times, the algorithm performs the flips in arbitrary order, yielding edge lengths $\ve{\tl}$ satisfying the triangle inequality only in the end. 

\paragraph{Hyperbolic Delaunay.} The reasons for applicability of \cref{eq:delaunay_ell} and use of \cref{eq:ptolemy-orig} are direct consequences of an elegant correspondence between hyperbolic and conformal metric structures used in the proofs of \cite{gu2018discrete,springborn2019ideal} and introduced in \cite{rivin1994euclidean}.  
We refer the readers to the detailed explanations in these papers and to the overview given in \cite[§5, §6]{Crane:2020:DCG}. Here, we state only three essential properties of the hyperbolic metric, without defining it explicitly, to provide some intuition for the algorithm's validity.
\begin{itemize}
    \item[(1)] For a given triangulation, a suitable hyperbolic metric is defined for any choice of edge lengths, not just choices of lengths satisfying the Euclidean triangle inequality.
    \item[(2)] The Delaunay property is well-defined for hyperbolic metrics, and can be tested using \cref{eq:delaunay_ell}. A flipped edge's length preserving the metric is given by the Ptolemy relation, and the classical flip algorithm is guaranteed to terminate.
    \item[(3)] If a triangulation is Delaunay in a hyperbolic metric, the corresponding Euclidean edge length assignment satisfies the triangle inequality and is also Delaunay. 
\end{itemize}

\section{Algorithm}
\label{sec:algo}

\begin{algorithm}[b]
\setstretch{0.9}
\SetAlgoLined
\DontPrintSemicolon
\SetKwInOut{Input}{Input}
\SetKwInOut{Output}{Output}
\SetKwProg{Fn}{Function}{:}{}
\SetKwRepeat{Do}{do}{while}
\SetKw{Not}{not}
\Input{
    triangle mesh $\cT = (V,E,F)$, closed, manifold\newline
    edge lengths $\ve{\ell} > 0$ satisfying triangle inequality\newline
    target angles $\ve{\hat\Theta} > 0$ respecting Gauss-Bonnet}
    \vspace{2pt}
\Output{
    triangle mesh $\cT' = (V,E',F')$\newline
    edge lengths $\ve{\ell}' > 0$ satisfying triangle inequality\newline
    such that $\|\ve{\Theta}_{(\cT',\ve{\ell}')} - \ve{\hat\Theta}\|_\infty \leq \varepsilon_\text{tol}$}
    \vspace{2pt}
\Fn{{\scshape{FindConformalMetric}}$(\cT, \ve{\ell}, \ve{\hat\Theta})$}
{
 $\ve{u} \gets \ve{0}$\;
 $(\cT, \ve{\ell}) \gets$ \scshape{MakeDelaunay}$(\cT, \ve{\ell}, \ve{u})$\;
 \While{\Not \scshape{converged}$(\cT, \ve{\ell}, \ve{u})$}{
   $\ve{g} \gets g(\cT, \ve{\ell}, \ve{u})$\tcp*{gradient}
   $H \gets H(\cT, \ve{\ell}, \ve{u})$\tcp*{Hessian}
   $\ve{d}\gets -H^{-1}\ve{g}$\tcp*{Newton direction}
   $(\cT, \ve{\ell},\ve{u})\gets$ \scshape{LineSearch}$(\cT, \ve{\ell}, \ve{u}, \ve{d})$\tcp*{Newton step}
 }
  $\ve{\ell}' \gets $\scshape{ScaleConformally}$(\cT, \ve{\ell}, \ve{u})$\;
 \Return $(\cT,\ve{\ell}')$\;
}

\vspace{4pt}

\Fn{\scshape{LineSearch}$(\cT, \ve{\ell}, \ve{u}, \ve{d})$}
{
    \While{true}{
    $(\cT, \ve{\ell}) \gets$ \scshape{MakeDelaunay}$(\cT, \ve{\ell}, \ve{u}+\ve{d})$\;  
   \If{$\langle\ve{d}, g(\cT, \ve{\ell}, \ve{u}+\ve{d})\rangle \leq 0$}
    { 
    \Return $(\cT, \ve{\ell},\ve{u}+\ve{d})$\;
    }
    $\ve{d} \gets \frac{1}{2}\ve{d}$\tcp*{backtracking line search}}
}

\vspace{4pt}

\Fn{\scshape{converged}$(\cT, \ve{\ell}, \ve{u})$}
{
  \Return $\|\ve{\hat\Theta} - \Theta(\cT,\ve{\ell}')\|_\infty \leq \varepsilon_\text{tol}$\;
}

\vspace{4pt}

\Fn{$g$ $(\cT, \ve{\ell}, \ve{u})$}
{
  \Return $\ve{\hat\Theta} - \Theta(\cT,\ve{\ell}')$\tcp*{\cref{eq:flattening1}}
}

\vspace{4pt}

\Fn{$H$ $(\cT, \ve{\ell}, \ve{u})$}
{
  \Return $\text{\scshape CotanLaplacian}(\cT,\ve{\ell}')$\;
}

\vspace{4pt}

\Fn(\tcp*[f]{angle computation}){$\Theta(\cT, \ve{\ell}, \ve{u})$}
{
  \For(\textbf{\tcp*[f]{\cref{eq:flattening1}}}){$v_i \in V$}
    {$\Theta_i \gets \sum_{T_{ijk}\in\cT'} \arccos{\left((\ell_{ij}'^2 + \ell_{ki}'^2 - \ell_{jk}'^2) / (2\ell'_{ij}\ell'_{ki})\right)}$\;}
  \Return $(\Theta_0,\dots,\Theta_n)$
}

\vspace{4pt}

\Fn{\scshape{MakeDelaunay}$(\cT, \ve{\ell}, \ve{u})$}
{
  \While{{\scshape{NonDelaunay}}$(\cT,\ve{\ell},\ve{u},e_{ij})$ \text{for any edge} $e_{ij}$}
    {$(\cT,\ve{\ell}) \gets \text{\scshape{PtolemyFlip}}(\cT,\ve{\ell}, e_{ij})$\;}
  \Return $(\cT,\ve{\ell})$\;
}

\vspace{4pt}

\Fn{\scshape{NonDelaunay}$(\cT, \ve{\ell}, \ve{u}, e_{ij})$}
{
  \Return $(\ell_{jk}'^2+\ell_{ki}'^2-\ell_{ij}'^2)/(\ell_{jk}'\ell_{ki}')\newline \phantom{xxx} +(\ell_{jm}'^2+\ell_{mi}'^2-\ell_{ij}'^2)/(\ell_{jm}'\ell_{mi}') < 0$\tcp*{\cref{eq:delaunay_ell}}
}

\vspace{4pt}

\Fn{\scshape{PtolemyFlip}$(\cT,\ve{\ell}, e_{ij})$}
{
  $\cT \gets$ \scshape{Flip}$(\cT, e_{ij})$\;
  $\ell_{km} \gets (\ell_{jk}\ell_{im}+\ell_{ki}\ell_{mj})/\ell_{ij}$\tcp*{\cref{eq:ptolemy-orig}}
  \Return $(\cT, \ve{\ell})$\;
}

 \caption{{\scshape FindConformalMetric}\label{alg:Newton}}
\end{algorithm}

As shown in \cite{gu2018discrete}, assuming an intrinsic Delaunay triangulation, $\ve{g}(\ve{u})$, see \cref{eq:flattening1}, is the gradient of a twice-differentiable, convex function $E(\ve{u}): \RV \rightarrow \mathbb{R}$, defined for arbitrary values of factors $\ve{u} \in \RV$.
Moreover, the  Hessian $H(\ve{u})$ of $E(\ve{u})$ is given by the standard discrete Laplacian (cotangent matrix, positive semi-definite) computed in the scaled metric given by $\ve{\ell}'=\ve{\ell}(\ve{u})$. Hence, the problem of finding the desired solution $\ve{u}^*$ reduces to the problem of minimization of a convex function, for which many methods with guaranteed convergence are known.

\paragraph{Newton's Method}
Newton's second-order optimization method is one suitable example. In \cref{alg:Newton} we spell it out, tailored to our purpose.
Most noteworthy is the fact that whenever $\ve{u}$ is updated (as initialization and during the line search), the triangulation is turned into a Delaunay triangulation with respect to the metric defined by~$\ve{u}$ through edge flipping. As discussed in \cref{sec:background}, this yields the same result as if one had performed edge flips in the specific sequence encountered during a continuous evolution of $\ve{u}$ from the previous to the updated state. Only then, values such as $\ve{g}(\ve{u})$ or $H(\ve{u})$ are computed on the mesh.

Note that while the function $E(\ve{u})$ is known explicitly, we entirely avoid using it in the algorithm. This is because the expressions that need to be evaluated for
$E(\ve{u})$ happen to be more complex and numerically less stable than the simple angle expressions needed for the gradient~$\ve{g}(\ve{u})$. Note that the use of $\ve{g}(\ve{u})$ in the line search condition is possible due to the convexity of $E(\ve{u})$.

\subsubsection*{Numerics}
The accuracy with which the target angles $\ve{\hat\Theta}$ can be matched of course depends (in a non-trivial manner) on the precision of the employed number type. If tolerance $\varepsilon_\text{tol}$ is chosen too low relative to this, \cref{alg:Newton} may never terminate. For practical purposes therefore additional stopping criteria can be taken into account: an upper bound on the number of Newton steps and the number of line search halvings, a lower bound on the Newton decrement $\langle\ve{d}, g(\cT, \ve{\ell}, \ve{u})\rangle$. Information about the practically achievable accuracy can be found in \cref{sec:res:prec}.

\section{Boundaries} 
\label{sec:boundary}

In the above we assumed $\cT$ to form a closed surface. For surfaces with boundary, we can reduce the problem to the case of closed 
surfaces. 

\subsection{Double Cover}
\label{sec:doublecover}

This reduction is achieved by means of a double cover approach:
\begin{enumerate}
    \item we attach a mirrored copy $N'$ of the input mesh $N$ along the boundary (merging boundary vertices and edges), as illustrated below, yielding a closed mesh $\cT$,
    \item we transfer the edge lengths $\ve{\ell}$ and the target curvatures $\kappa_i$ of interior vertices $v_i$ from $N$ to $N'$, 
    \item we prescribe $\hat\Theta_i = 2\pi -2\hat\kappa_i$ at each (former) boundary vertex $v_i$, where $\kappa_i$ is the target discrete geodesic boundary curvature at vertex $v_i$.
\end{enumerate}

\begin{wrapfigure}[8]{r}{0.25\linewidth}
\vspace{-4.0mm}
\hspace{-0.5cm}
\begin{overpic}[width=1.2\linewidth]{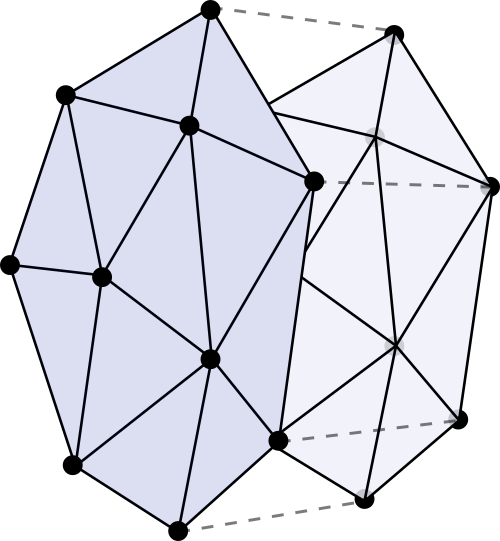}
\end{overpic}
\end{wrapfigure}
The double cover mesh $\cT$ built this way exhibits an obvious
reflectional symmetry, i.e., there is a map $R$ with $R^2 = I$ that takes vertices to vertices, edges to edges, and faces to faces. It maps an element stemming from the interior of $N$ to its copy in $N'$ and vice versa; on the merged (former) boundary vertices and edges, $R$ is the identity.

\paragraph{Conformal Metric Symmetry}
Due to symmetry (invariance with respect to $R$) of the mesh $\cT$, the metric $\ve{\ell}$, and the target angles $\ve{\hat\Theta}$, the symmetrically initialized factors $\ve{u}$ will remain symmetric after each iteration of the  optimization process (up to round-off error).  This can be seen by observing that the function $E(\ve{u})$ is the sum of per-triangle terms $E_T(\ve{u}_T)$, where $\ve{u}_T$ is the restriction of $\ve{u}$ to vertices of the triangle $T$. 
Given the above symmetry, its gradient $\ve{g}(\ve{u}) = \nabla_{\ve{u}} E$ therefore is invariant with respect to $R$.
As a consequence, the metric symmetry is maintained as $\ve{u}$ evolves, and the triangulation remains symmetric during evolution.
Therefore the resulting conformal metric will be symmetric, i.e., identical on the two copies, as well. Consequently, if we cut the mesh along the symmetry line (the former boundary) in the end, so as to discard one copy, a boundary vertex $v_i$ will have exactly half the prescribed angle, $\frac{1}{2}\hat\Theta_i = \pi-\hat\kappa_i$, and therefore exhibit a discrete geodesic boundary curvature of $\hat\kappa_i$, just as intended.
The fact that $\ve{u}$ (and thus all vertex-associated attributes) evolve symmetrically furthermore implies that we can use a \emph{tufted} double cover as in \cite{sharp2020laplacian}, where not only boundary vertices but all vertices are shared between the two symmetric halves of $\cT$. This reduces the number of variables in the optimization problem.
This symmetry does not mean, however, that computations could trivially be restricted entirely to one half of the double cover only: edge flips may, and commonly will, create edges and faces spanning both halves of the double cover, crossing the symmetry line. 
Even more importantly, the symmetry leads to
co-circular vertex configurations that are \emph{stable}, i.e., for a given triangulation these remain co-circular \emph{independent of the evolution of $\ve{u}$}. These configurations need to be handled specially, as we explain in more detail below.

\subsection{Symmetric Meshes}

\newcommand{\Ea}{E^\parallel}
\newcommand{\Ep}{E^\perp}
\newcommand{\Hp}{E^\perp}

We begin by making precise the notion of \emph{combinatorially symmetric} polygon mesh.
    In this, rather than using edges, we use \emph{halfedges}, each  
    associated with a unique face (or boundary loop, which can be treated exactly like a face). Specifically, each edge 
    corresponds to two halfedges. We will make use of this abstract notion of symmetric meshes to ensure that our subsequent considerations of edge flips and related aspects cover all combinatorial cases that may occur.

    \newcommand{\cN}{{\mathcal{N}}}
    \newcommand{\cO}{{\mathcal{O}}}
    \newcommand{\cC}{{\mathcal{C}}}
    \newcommand{\cQ}{{\mathcal{Q}}}

    \begin{defn}[Combinatorial Mesh]
    A combinatorial polygon mesh is a triple $(H,\cN,\cO)$ of a set of halfedges $H$, a bijective function $\cN: H \rightarrow N$ (\emph{next-halfedge} function), and a bijective function $\cO$ (\emph{opposite-halfedge} function) with the property
    \begin{equation}
    \cO^2(h) = h;\; \cO(h) \neq h
     \label{eq:O-prop}
    \end{equation} 
    i.e., all orbits of $\cO$ have size 2.
    \end{defn}

    \begin{defn}[Mesh Elements]
    Define the bijective circulator
    function $\cC: H \rightarrow H$ to be $\cN^{-1}(\cO(h))$. Then the mesh has the following implied elements:
    \begin{itemize}
        \item \emph{Faces} are the orbits of the next-halfedge function $\cN$.
        \item \emph{Vertices} are the orbits of the circulator function $\cC$.
        \item \emph{Edges} are the orbits of the opposite-halfedge function $\cO$.
    \end{itemize}
    Collectively we refer to them as (mesh) \emph{elements}. A halfedge \emph{belongs} to an element if it is part of the respective orbit.
    \end{defn}

    A \emph{mesh with boundary} is a mesh with a subset of its faces marked as boundary loops. The halfedges of these loops form the set $H^{bnd}$ of \emph{boundary halfedges}.

\begin{defn}[Reflection Map]
\label{def:R}
A reflection map $R: H \rightarrow H$
for a mesh $(H,\cN, \cO)$ without boundary is an involution ($R^2 = I$) defined on the set of halfedges: each halfedge is mapped either to itself, or forms a reflection pair with a distinct halfedge.
It is required to satisfy the following conditions: 
\begin{enumerate}
    \item preservation of $\cO$ relation: $\cO(R(h)) = R(\cO(h))$,
    \item inversion of $\cN$ relation: $\cN(R(h)) = R(\cN^{-1}(h))$,
    \item preservation of boundary: $h \in H^{bnd} \iff R(h) \in H^{bnd}$.
\end{enumerate}
\end{defn}
Note that conditions (1) and (2) correspond to the properties of continuity and orientation-reversal of continuous reflection maps \cite{FieldsSymSurfaces}.
They imply that $R$ maps orbits of $\cN$, of $\cO$, and, as a consequence, of $\cC$, to orbits of these functions, i.e., it is well-defined for faces, edges, and vertices (via $R(x) = x' \iff R(h) \in x'$ for any $h\in x$). 
Furthermore, because $R^2 = I$, all orbits of $R$
have length 1 or 2, whether it acts on halfedges, faces, edges, or vertices.
This implies the following partitioning.

\begin{prop}[Halfedge Sets]
$H$ can be partitioned into disjoint sets $H^1$, $H^2$, $H^s$  so that the following conditions are satisfied: 
\begin{itemize}
    \item $h \in H^s \iff R(h) = h$;
    \item $h \in H^1 \iff R(h) \in H^2$;
    \item for any face or edge $x$ either all belonging halfedges are in $H^1$, 
    or all are in $H^2$, or it is fixed by $R$ (i.e. $R(x)=x$)
\end{itemize}
\end{prop}
\begin{proof} 
If $x$ is not fixed, by the well-definedness of $R$ on mesh elements, for each $h\in x$ we have $R(h) \not\in x$. Therefore for a non-fixed individual face or edge $x$ all its halfedges can be assigned to $H^1$ (or to $H^2$) without contradicting the above conditions. It needs to be shown that this can be done for all such elements consistently.

Let $H^e$ the set of halfedges whose edges are not fixed and $H^f$ the set of halfedges whose faces are not fixed.
Let $\cQ$ the relation that is the union of $\cO|_{H^e}$ and $\cN|_{H^f}$ on $H\setminus H^s$.
Consider the connected components $H_i$ of $\cQ$ (intuitively: the mesh's connected components separated by fixed edges and fixed faces). Due to the properties of $R$ (preserving/inverting $\cO$ and $\cN$) it is well-defined on these connected components via $R(H_i) = H_j \Leftrightarrow R(h) \in H_j$ for any $h \in H_i$. Using arguments analogous to \cite[Prop.~2]{FieldsSymSurfaces} one verifies that the set of fixed elements necessarily forms a cycle; therefore there are at least two such connected components.

As $R$ on $H\setminus H^s$ has orbits of length 2 only, it allows a bipartition of the connected components, i.e., they can be assigned to two sets $H^1$ and $H^2$ in accordance with the above conditions.
\end{proof}

This leads to the following partitioning of 
the sets of edges and faces, where $e = (h,h')$, $f = (h_0,\ldots h_{m-1})$ denote the orbits of belonging halfedges:

\begin{itemize}
\item \makebox[.8cm]{$e  \in E_i$\hfill} $\iff h,h' \in H^i$, $i=1,2$
\item \makebox[.8cm]{$e \in \Ep$\hfill} $\iff h,h' \in H^s$
\item \makebox[.8cm]{$e \in \Ea$\hfill} $\iff h = R(h')$
\item \makebox[.8cm]{ $f \in F_i$\hfill} $\iff h_0 \in H^i$, $i=1,2$
\item \makebox[.8cm]{$f \in F^s$\hfill} $\iff R(h_0) \in f$
\end{itemize}
The set $\Ep$ is the set of edges (perpendicularly) crossing the symmetry line between two halves of a symmetric mesh mapped to each other (see \cref{fig:symmflips} right); the set  $\Ea$  is the set of edges on the symmetry line; $F^s$ is the set of faces that cross the symmetry line, and are mapped by the symmetry map to themselves.

\subsubsection{Double Cover Construction}

Using the terminology established above, our construction from \cref{sec:doublecover} can be formalized as follows.
Given a mesh $N = (H^0,\cN^0, \cO^0)$,
with boundary and interior halfedges $H^{bnd} \cup H^{int} = H^0$, we discard $H^{bnd}$ and set $H = H^1\cup H^2$ where $H^1 = H^{int}$ and $H^2 = \bar H^{int}$, where $\bar{\cdot}$ denotes a copy. The reflection map $R$ is defined via $R(h) := h'$ if $h'\in H^2$ is the copy of $h\in H^1$.
$\cO^0$ is adopted on both copies to define $\cO$, except that $\cO(h) := R(h)$ if $\cO^0(h)\in H^{bnd}$; this latter adjustment constitutes the \emph{gluing} of the two copies along their boundaries.
Finally
\begin{align*}
\cN(h) :=
\begin{cases}
\cN^0(h) &\text{ if } h\in H^1,\\
R({\cN^{0}}^{-1}(R(h))) &\text{ if } h\in H^2.
\end{cases}
\end{align*}

This forms the symmetric double cover mesh $\cT=(H, \cN, \cO, R)$ with triangle faces and map $R$. Note that $R$ is a reflection map: it satisfies the conditions of \cref{def:R} (where condition (3) is void as $\cT$ has no boundary). It is easy to see that this construction implies $\Hp =\varnothing$ and $F^s = \varnothing$, i.e., no element crosses the symmetry line (the former boundary). $\Ea$ contains the edges lying \emph{on} the symmetry line, i.e., those for whose halfedges the $\cO$ relation was adjusted to glue the two copies.

This initially simple situation can change, however, when edge flips are performed on the double cover mesh.

\subsection{Symmetric Flips}

When an edge $e$ in a symmetric mesh $M = (H,\cN,\cO,R)$ shall be flipped, the edge $R(e)$ needs to be flipped as well (unless $R(e) = e$), so as to be able to maintain a symmetric mesh. The simultaneous flip of $e$ and $R(e)$ (as well as the single flip of $e$ if $R(e) = e$) is referred to as \emph{symmetric flip}. As discussed in \cref{sec:doublecover}, in the algorithm from \cref{sec:algo} the metric evolves symmetrically. This implies that whenever the algorithm intends to flip an edge $e$, it simultaneously intends to flip $R(e)$ as well. The algorithm is therefore compatible with the restriction to symmetric flips.

While for an edge $e\in E_i$ with incident faces $f,g\in F_i$ the process is obvious, special care needs to be taken when elements from $\Ea$, $\Ep$, or $F^s$ are involved. We will exhaustively distinguish different types of symmetric flips based on the membership of the involved edges and faces in these sets.

\paragraph{Flip Types}
For a triple $(f_a, e, f_b)$ of an edge $e$ with incident faces $f_a$, $f_b$, the triple of labels denoting their set memberships, e.g., $(1,\parallel,2)$, is called flip \emph{type} of the edge $e$.

\paragraph{Consistent Flip Types}
We say that a type is \emph{consistent} if it may occur in a symmetric mesh. For instance, $(1,\perp,1)$ is not a consistent type, as edges from $\Ep$ necessarily have incident faces from $F^s$ by definition.
The following statements help ruling out combinations of labels for a triple $(f_a, e, f_b)$ in a symmetric mesh (up to exchange of $f_a$ and $f_b$):\bigskip\bigskip\bigskip

\begin{prop}[Label Compatibility] \quad
\begin{itemize}
\item[(a)] $e\in\Ep \Rightarrow f_a,f_b\in F^s$.
 
\item[(b)] $e\in\Ea \Rightarrow f_a\in F^1, f_b\in F^2$ or $f_a=f_b\in F^s$.

\item[(c)] $e\in E^1 \Rightarrow f \notin F^2$,\; $e\in E^2 \Rightarrow f \notin F^1$.

\item[(d)] $e\in E^i$, $f_a,f_b\in F^s \Rightarrow R(e)\in f_a,f_b$.
\end{itemize}
\label{prop:edge-faces1}
\end{prop}
\begin{proof}
Part (a) follows immediately from the definition of $F^i$, as faces from $F^i$ cannot have edges from $\Ep$.

Suppose a face $f_a$ is incident at an edge $e$ from $\Ea$. For these edges $R(e) = e$. Suppose $f_a\in F^1$, then $R(f_a)$ is incident to $R(e) = e$, therefore $f_b = R(f_a)$. As $R(f_a) \in F^2$ by definition of $F^2$, this proves the first part of (b). Suppose $f_a \in F^s$, and let $h$ a halfedge $h \in e$, $h \in f_a$. Then $R(h) \in f_a$ by the 
definition of $F^s$; but, by definition of $\Ea$, $R(h)\in e$, so $f_a = f_b$, i.e., a face is adjacent to itself along $e$. 

Part (c) directly follows from the definitions of 
$E^i$ and $F^i$. 

In part (d), 
suppose $f_a$ and $f_b$ are incident at $e\in E^1$, $f_a, f_b \in F^s$, and $e=(h_a,h_b)$. 
Then $R(h_a) \in R(f_a) = f_a$,  $R(h_b) \in f_b$, and $\cO(R(h_a)) = R(h_b)$ by the 
properties of $R$, i.e., $(R(h_a), R(h_b))$
is an edge. By definition of $E^i$, it has to be in $E^2$, i.e., faces $f_a$ and $f_b$ share a second edge, and this edge is from $E^2$. 
\end{proof}

\cref{prop:edge-faces1} leaves the following six possibilities, up to a $1 \leftrightarrow 2$ exchange. It is easy to construct examples proving that all of them are consistent, i.e., may occur in a symmetric mesh:
\begin{itemize}
    \item \makebox[1.8cm]{Edge in $E^1$:\hfill} 
    \makebox[3.2cm]{Set 1a: $(1,1,1)$, $(1,1,s)$\hfill}
    Set 1b: $(s,1,s)$
    \item \makebox[1.8cm]{Edge in $\Ea$:\hfill}
    \makebox[3.2cm]{Set 2a: $(1,\parallel,2)$\hfill}
    Set 2b: $(s,\parallel, s)$
    \item \makebox[1.8cm]{Edge in $\Ep$:\hfill}  Set 3: \;\,$(s,\perp,s)$
\end{itemize}

\begin{figure}[b]
    \centering
    \begin{overpic}[width=.97\linewidth]{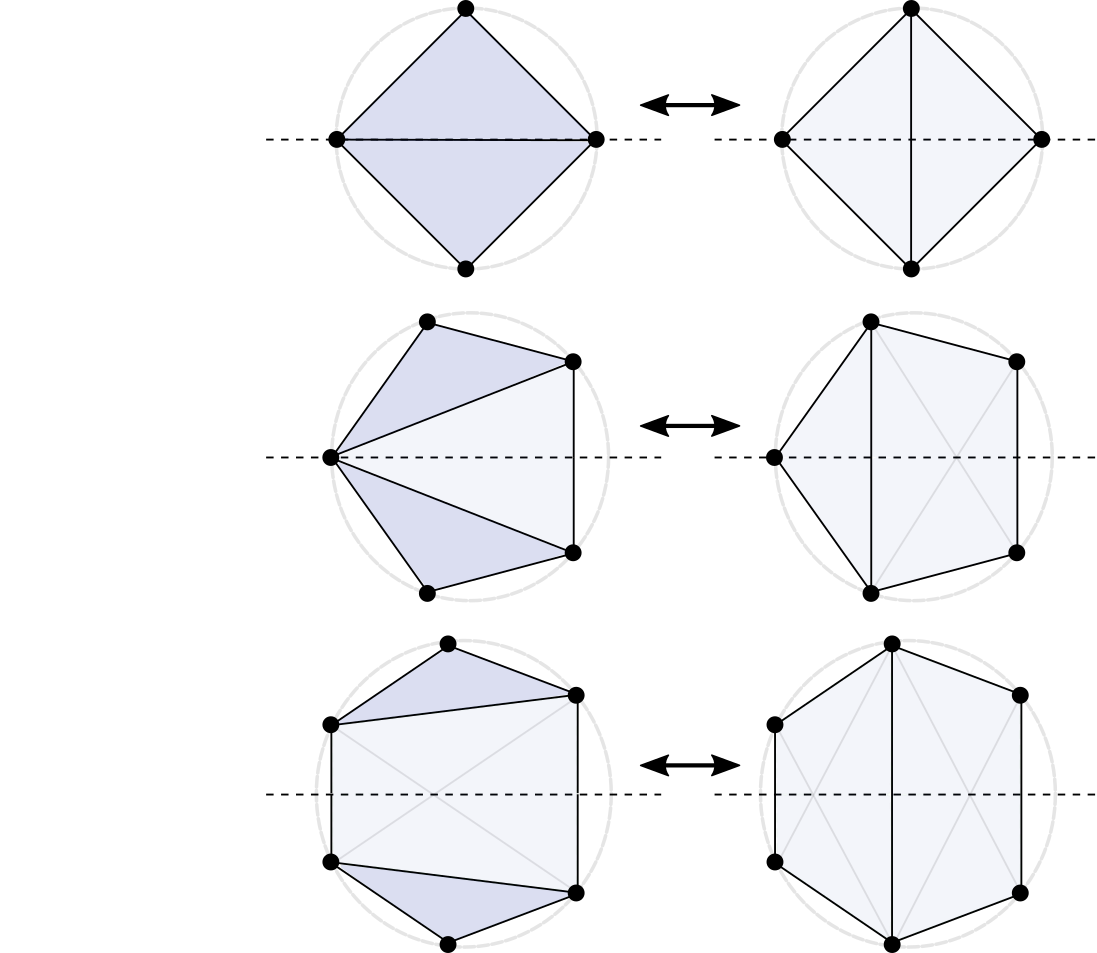}
    \small
    \put(2.5,73.0){
        \begin{tabular}{c}
            $(1,\parallel,2)$\\
            $\leftrightarrow$\\
            $(t, \perp, t)$
        \end{tabular}
    }
    
    \put(-4,44.5){
        \begin{tabular}{c}
            $(1,1,t) + (2,2,t)$\\
            $\leftrightarrow$\\
            $(t, \perp, q)$
        \end{tabular}
    }
    
    \put(-4,13.7){
        \begin{tabular}{c}
            $(1,1,q) + (2,2,q)$\\
            $\leftrightarrow$\\
            $(q, \perp, q)$
        \end{tabular}
    }
    \end{overpic}
    \caption{Symmetric edge flips involving faces from $F^s$ (light blue), crossing the symmetry line (dashed). Faces from $F^1$ and $F^2$ are colored dark blue. The configurations are shown with co-circular vertices, though combinatorially flips can be performed in any state. Note that the light blue quads' vertices, however, are necessarily co-circular by symmetry, regardless of metric.}
    \label{fig:symmflips}
\end{figure}

\begin{table*}[t]
\caption{Combinatorial updates required to perform symmetric flips of all relevant consistent types. The change to $\cN$ is given by listing the orbits (halfedge cycles forming faces) of $\cN$ created by the flip. The employed indexing is depicted in the figures left and right.
Similarly, we define changes to $R$ viewing it as a permutation with orbits of length 1 or 2, and listing the sets of orbits being replaced. Finally, rather than deleting and adding new halfedges on demand, for implementational efficiency we can associate a superfluous pair of halfedges, eliminated by a quad-creating flip, with the quad (listed behind the bar). Because a flip that requires a new pair of halfedges always eliminates a quad, this pair can then be reused. \label{tab:flips}}
\renewcommand{\arraystretch}{1.2}
\small
\begin{tabular}{ccc}
\raisebox{-.5\height}{
    \begin{overpic}[trim=0 0 380 0,clip,width=.088\linewidth]{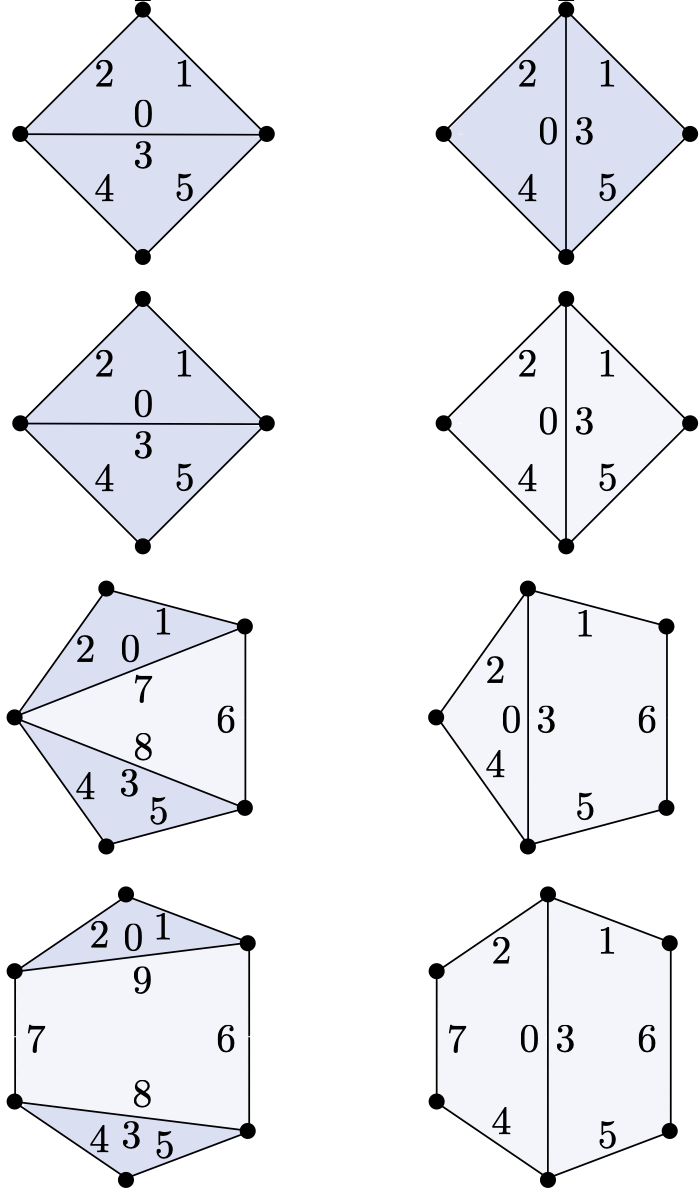}
    \small
    \end{overpic}}\;\,
    &
\begin{tabular}{|l|l|}
\hline
 \multicolumn{2}{|c|}{$(1,1,1)+(2,2,2)\;\;\leftrightarrow\;\; (1,1,1)+(2,2,2)$}  \\\hline
      $\cN:(h^i_{0},h^i_{1},h^i_{2}),\; (h^i_{3},h^i_{4},h^i_{5}),\; i=1,2$& 
      $\cN:(h^i_0,h^i_2,h^i_{4}),\; (h^i_{1},h^i_{3},h^i_5),\; i=1,2$\\
     $R:$ unchanged & $R:$ unchanged\\\hline
 \multicolumn{2}{c}{\vspace{-0.25cm}}  \\\hline
 \multicolumn{2}{|c|}{$(1,\parallel,2)\;\;\leftrightarrow\;\; (t,\perp,t)$} \\\hline
     $\cN: (h_{0},h_{1},h_{2}),\; (h_{3},h_{4},h_{5})$& 
     $\cN: (h_{0},h_{2},h_{4})$,\; $(h_{1},h_{3},h_{5})$\\  
     $R: (h_{0},h_{3})$ & $R: (h_{0}),\; (h_{3})$\\\hline
 \multicolumn{2}{c}{\vspace{-0.25cm}}  \\\hline
  \multicolumn{2}{|c|}{$(1,1,t)+(2,2,t)\;\;\leftrightarrow\;\; (t,\perp,q)$\phantom{xxxxxxx.}} \\\hline
     $\cN: (h_{0},h_{1},h_{2}),\; (h_{3},h_{4},h_{5}),(h_{6}, h_{7},h_{8})$& 
     $\cN: (h_{0},h_{2},h_{4}),\; (h_{1},h_{3},h_{5},h_{6}) \,\;|\;\, h_{7},h_{8}$\\
     $R: (h_{0},h_{3}),\; (h_{7},h_{8})$ & $R: (h_{0}),\; (h_{3})$\\\hline
 \multicolumn{2}{c}{\vspace{-0.25cm}}  \\\hline
 \multicolumn{2}{|c|}{$(1,1,q)+(2,2,q)\;\;\leftrightarrow\;\; (q,\perp,q)$\phantom{xxxxxxxx}} \\\hline
     $\cN: (h_{0},h_{1},h_{2}),\; (h_{3},h_{4},h_{5}),(h_{6}, h_{9},h_{7},h_{8})$\qquad\qquad\qquad& 
     $\cN: (h_{0},h_{2},h_{7}, h_{4}),\;  (h_{1},h_{3},h_{5},h_{6}) \,\;|\;\, h_{8},h_{9}$\qquad\qquad\qquad\\
     $R: (h_{0},h_{3}),\; (h_{8},h_{9})$ & $R: (h_{0}),\; (h_{3})$\\\hline
\end{tabular}
&
\raisebox{-.5\height}{
    \begin{overpic}[trim=380 0 0 0,clip,width=.088\linewidth]{images/labeled_all2.png}
    \small
    \end{overpic}}
\end{tabular}
\vspace{0.2cm}
\end{table*}

\paragraph{Relevant Flip Types} Among these types, only four are also \emph{relevant}; we show in \cref{sec:irrel} that the two types in the sets 1b and 2b are necessarily associated with edges that satisfy the Delaunay condition \cref{eq:delaunay_ell} irrespective of the choice of lengths of edges involved. These are not relevant for the purpose of the algorithm from \cref{sec:algo}, which exclusively flips non-Delaunay edges. This leaves only sets 1a, 2a, and 3 for further consideration.

\paragraph{Triangles and Quadrilaterals} 
As we shall see, a flip of type $(1,1,s)$ leads to a
pair of triangles in $F^s$ that together form a quadrilateral which is inscribed, i.e., the four vertices are intrinsically co-circular (\cref{fig:symmflips}). Remarkably, this statement holds regardless of metric, as long as it is symmetric, i.e., invariant with respect to $R$. Instead of randomly choosing a diagonal splitting this quadrilateral into two triangles, we explicitly represent it as a quadrilateral face. This avoids violating the symmetry by the diagonal, which would complicate recovering the surface with boundary after the conformal metric is computed, and avoids issues such as infinite sequences of flips caused by stably (numerically nearly) co-circular points. 

Faces in $F^s$ can therefore be triangular or quadrilateral. We accordingly partition $F^s = F^t \cup F^q$, and further distinguish $t$-versions and $q$-versions of flip types involving the label $s$. This yields a total of seven types that are consistent and relevant. They are related as follows by the fact that a symmetric flip of one or two edges of certain types reversibly yields a configuration of different type:
\begin{enumerate}
\item \makebox[2.9cm]{\hfill$(1,1,1) + (2,2,2)$} $\;\leftrightarrow\; (1,1,1) + (2,2,2)$;
\item \makebox[2.9cm]{\hfill$(1,\parallel,2)$} $\;\leftrightarrow\; (t, \perp, t)$;
\item \makebox[2.9cm]{\hfill$(1,1,t\;\!) + (2,2,t\;\!)$} $\;\leftrightarrow\; (t, \perp, q)$;
\item \makebox[2.9cm]{\hfill$(1,1,q) + (2,2,q)$} $\;\leftrightarrow\; (q, \perp, q)$. 
\end{enumerate}
Case (1) is the standard case of flipping a configuration not involving the symmetry line. (2), (3), and (4) are the special cases crossing the symmetry line; they are illustrated in \cref{fig:symmflips}. \cref{tab:flips} details the combinatorial changes to be performed on the symmetric mesh so as to execute these symmetric flips.

\subsection{Symmetric Metric} 
We now assume the symmetric combinatorial mesh $\cT$ is equipped with a metric, as in the algorithm from \cref{sec:algo}, and that this metric is symmetric as well.

\paragraph{Delaunay Criterion}
For edges with two incident triangles, the Delaunay check needed for the algorithm is standard, via~\cref{eq:delaunay_ell}.
If one of the incident faces is a quad, due to symmetry it, regardless of the metric, is an inscribed trapezoid. As a consequence, whichever way we (virtually) split it into triangles we get the same angles opposite any of its edges. Hence, we may perform the Delaunay check based on arbitrary virtual diagonals in the quads. 

\paragraph{Gradient and Hessian}
For the same reason, the computation of gradient $g(\ve{u})$ and Hessian $H(\ve{u})$ can be performed based on arbitrary diagonals; the choice does not affect the result \cite{springborn2019ideal}.

\paragraph{Ptolemy Formula}
Note that each of the edges created by symmetric flips involving quads (\cref{fig:symmflips}) can also be obtained by a sequence of edge flips involving triangles (and split quads) only. In this way the length of such edges can be computed using (multiple instances of) the standard Ptolemy formula \cref{eq:ptolemy-orig}. As there are only four types of flips involving quads, one can conveniently derive closed form expressions for these cases in advance, rather than actually performing these sequences for each flip. Note that each quad needs to store its diagonal length to enable these computations.

\subsubsection{Irrelevance of Flip Types $(s,\parallel,s)$ and $(s,1,s)$}
\label{sec:irrel}
\begin{prop}
Types 
$(t,\parallel,t)$, $(q,\parallel,q)$, 
$(t,1,t)$, $(t,1,q)$, and $(q,1,q)$
are associated with edges that are Delaunay regardless of metric.
\label{prop:del-configs}
\end{prop}
\paragraph{\emph{\textsc{Proof}}}
Consider $(t,\parallel,t)$. By \cref{prop:edge-faces1}(b), it corresponds to a configuration 
with a single face:  $(f^t, e^\parallel, f^t)$. As the triangle $f^t$ is 
isosceles, and both side edges of the 
triangle coincide with $e^\parallel$, 
angles opposite $e^\parallel$ are $\pi/2-\alpha/2$ if the apex angle is $\alpha$, i.e., their sum is guaranteed 
to be less than $\pi$ and the edge is Delaunay. 
For $(q,\parallel,q)$, to evaluate
\begin{wrapfigure}[8]{r}{0.15\linewidth}
\vspace{-2.0mm}
\hspace{-0.5cm}
\begin{overpic}[width=1.2\linewidth]{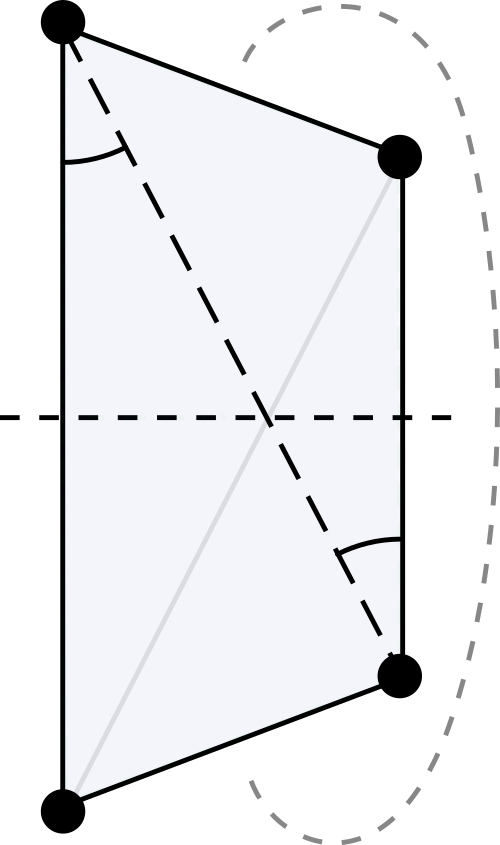}
\put(23,14){$e^\parallel$}
\put(23,75.5){$e^\parallel$}
\end{overpic}
\end{wrapfigure}the 
Delaunay criterion, we split 
$f^q$ into triangles. As $f^q$ is inscribed the choice of diagonal does not affect the angles; we can choose the diagonal that 
connects a vertex of $e^\parallel$ with 
a vertex with trapezoid angles $\leq\pi/2$
(see inset), from which we can see that both
angles opposite $e^\parallel$ are less than $\pi/2$. 
For cases $(t,1,t)$, $(t,1,q)$, and $(q,1,q)$ the same logic applies to each 
face incident at the shared edge~$e^1$.\qed

\subsection{Restriction to Single Cover}

Once the algorithm from \cref{sec:algo} has terminated and the desired conformal metric has been computed, we finally need to discard half of the double cover---or transfer the metric from half of it onto the original surface, depending on implementation specifics. To this end we effectively need to cut the symmetric surface along the line of symmetry. While initially the entire symmetry line is formed by a sequence of mesh edges, this may no longer be the case due to flips, namely whenever $F^s$ and $\Ep$ are not empty in the end. One simply needs to split all edges from $\Ep$ at their midpoint, and split the triangles and quads from $F^s$ by connecting these inserted split vertices. Alternatively, if an overlay data structure \cite{fisher2007algorithm} is used to keep track of a bijection between original mesh and modified mesh, the restriction to one half of the double cover is even easier, as the original edges are (as a whole or in parts) still present in the overlay mesh.

\section{Evaluation}

We have implemented \cref{alg:Newton} (with support for boundaries following \cref{sec:boundary}) in \texttt{C++}. Due to the method's solid theoretical foundation, the only limitation is due to numerical precision limits; to be able to assess this aspect, the implementation supports the optional use of the MPFR multi-precision floating point number type instead of standard double precision numbers, enabling the variation of numerical precision (as done in \cref{sec:res:prec}).

\subsection{Validation}
\subsubsection*{Closed Surfaces}

A dataset of mesh models together with angle prescriptions $\ve{\hat\Theta} > 0$ has been released with \cite{Myles:2014}. We applied our implementation to the closed models from this dataset; the error decay in the course of the algorithm on these cases is visualized in \cref{fig:decayMPZ}.

As further test instances we use 1000 different random target angle prescriptions $\ve{\hat\Theta}$ (with $\hat\Theta_i\in (\pi,3\pi)$ for all vertices $v_i$) on a sphere mesh (1K vertices). The error decay is visualized in \cref{fig:decay}. Note that the overall behavior is very similar, whether the prescribed angles are random or geometrically meaningful (as in \cref{fig:decayMPZ}).

We consider the extreme scenario of concentrating the target metric's entire curvature in one point (i.e., prescribing a single cone of angle $2\pi(2g-1)$ in an otherwise flat metric). Results for surfaces of increasing genus $g$ (procedurally generated $g$-tori) are shown in \cref{fig:genus}. A blow-up of the situation around the prescribed cone vertex on the genus 10 example is shown in \cref{fig:res:single}.

\begin{figure}[h]
    \centering

\begin{overpic}[scale=1.0]{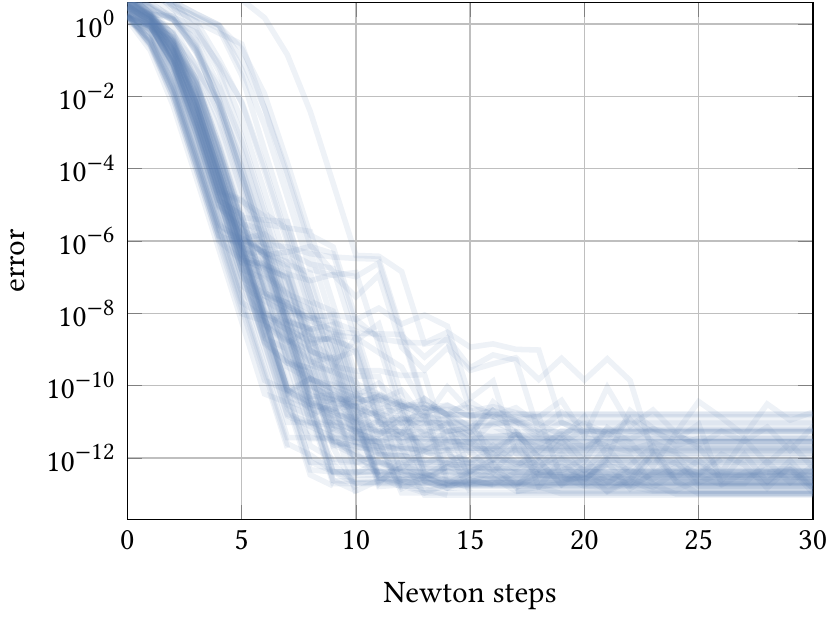}
  \end{overpic}

\vspace{-0.3cm}
    \caption{Decay of maximum angle error $\|\ve{\hat\Theta} - \ve{\Theta}\|_\infty$ over the iterations of the Newton algorithm. Each graph represents one of the instances from the dataset of \cite{Myles:2014}. 
    }
    \label{fig:decayMPZ}
    \vspace{0.5cm}
\end{figure}

\begin{figure}[h]
    \centering

 \begin{overpic}[scale=1.0]{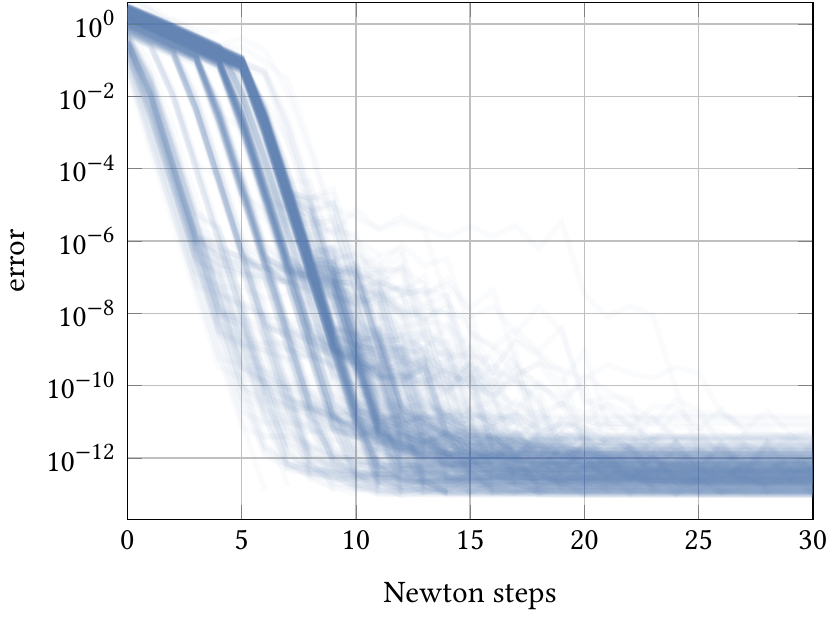}
     \put(14.5,11.58){\includegraphics[scale=1.0,trim = 35 28 0 0,clip]{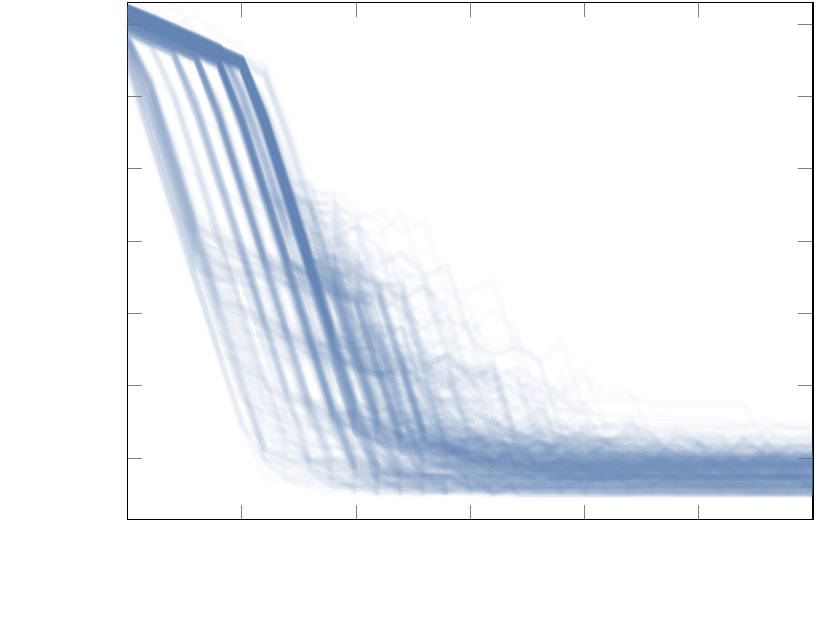}}  
  \end{overpic}

\vspace{-0.3cm}
    \caption{Decay of maximum angle error $\|\ve{\hat\Theta} - \ve{\Theta}\|_\infty$ over the iterations of the algorithm. Each graph represents one of 1000 random test instances. 
    }
    \label{fig:decay}
    \vspace{0.6cm}
\end{figure}

\begin{figure}[h!]
    \centering
    
 \begin{overpic}[scale=1.0]{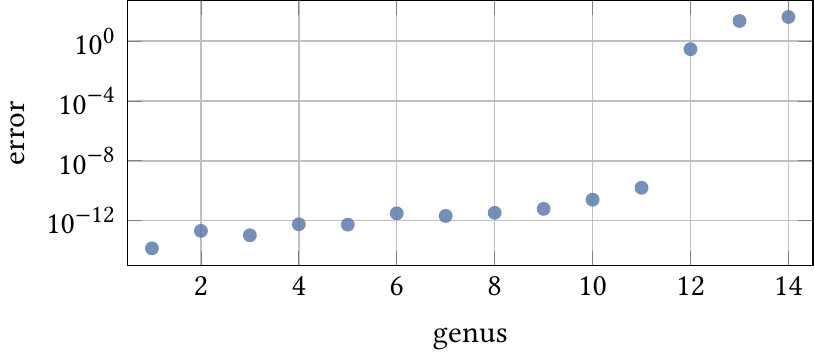}\end{overpic}
  
\vspace{-0.3cm}
    \caption{Final residual angle error for the extreme case of concentrating all curvature in a single cone on an $g$-torus surface (genus $g$). For the genus 11 case, where the residual error is still benign, the conformal scale factor spans 87 orders of magnitude. For the problematic genus 12 case it surpasses 100.
    }
    \label{fig:genus}
    \vspace{-0.9cm}
\end{figure}

\begin{figure}[t]
    \centering
 \begin{overpic}[width=0.99\linewidth]{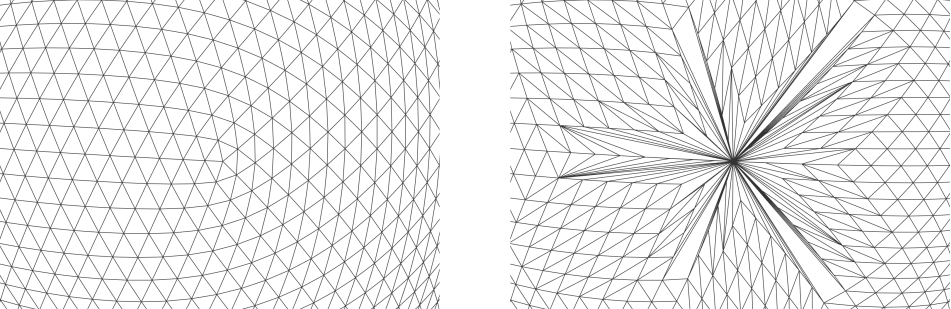}
  \end{overpic}
  \vspace{-0.15cm}
    \caption{Triangulation around a vertex with prescribed angle $\Theta = 38\pi$, before and after executing the algorithm. The triangulation on the right, while exhibiting bad angles under the depicted original metric, is Delaunay under the computed conformal metric (with curvature $-36\pi$ at the central vertex).}
    \label{fig:res:single}
\end{figure}

\subsubsection*{Surfaces with Boundary}

Similar to the experiment for surfaces without boundary, we generate 1000 different random target angle prescriptions $\ve{\hat\Theta}$ for a surface with boundary (a disk with 5K vertices). In the interior we prescribe a flat metric, at the boundary we prescribe a geodesic curvature, maximally in the range $\pm\pi$, i.e., $\hat\Theta_i\in (0,2\pi)$ for all boundary vertices $v_i$.
\cref{fig:boundarydata} shows the number of the different types of symmetric flips that are performed in the course of the algorithm on these cases. As expected, the number of flips is larger for cases with a prescribed curvature spanning a larger range.

\begin{figure}[b]
    \centering
 \begin{overpic}[scale=1.0]{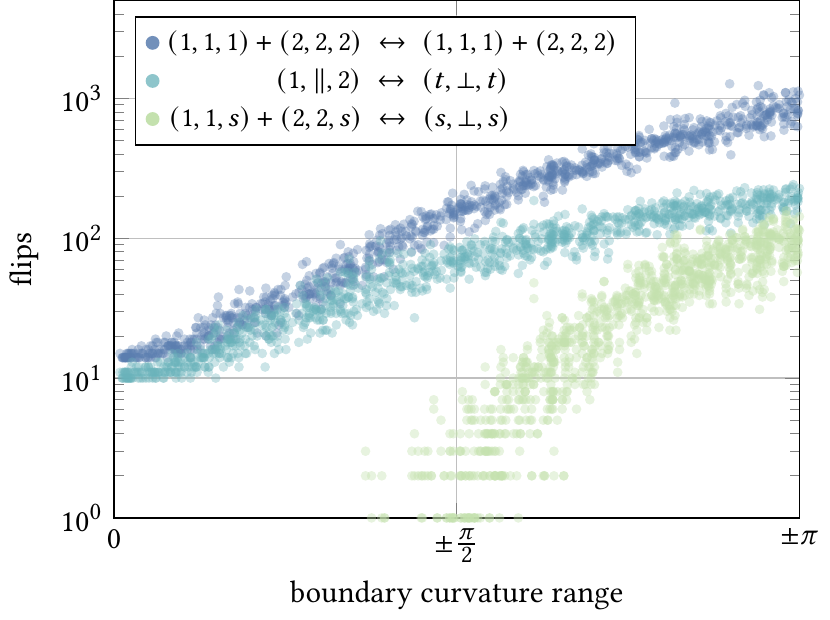}\end{overpic}
\vspace{-0.3cm}
    \caption{Scatter plot showing the numbers of different types of symmetric flips during the algorithm relative to the range of prescribed random boundary curvatures. Each dot represents one type of flips for one of 1000 test instances.
    }
    \label{fig:boundarydata}
\end{figure}

The above mentioned dataset from \cite{Myles:2014} also contains meshes with one or more boundary loops, together with angle prescriptions $\ve{\hat\Theta} > 0$ for interior and boundary vertices. The error decay on these cases is shown in \cref{fig:decay:MPZboundary}.

\begin{figure}[h]
    \centering
    
 \begin{overpic}[scale=1.0]{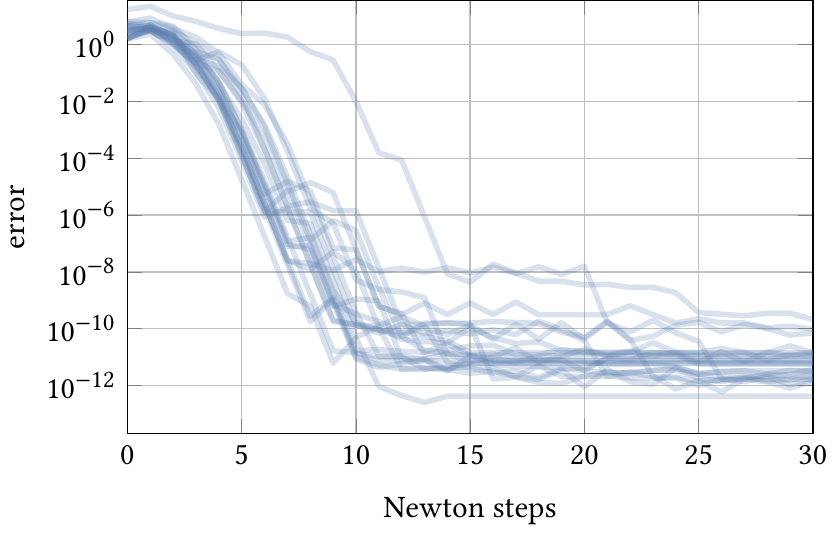}\end{overpic}

\vspace{-0.3cm}
    \caption{Decay of maximum angle error $\|\ve{\hat\Theta} - \ve{\Theta}\|_\infty$ over the iterations of the Newton algorithm. Each graph represents one of the instances \emph{with boundary} from the dataset of \cite{Myles:2014}. 
    }
    \label{fig:decay:MPZboundary}
    \vspace{0.3cm}
\end{figure}

\subsection{Comparison}
\label{sec:res:comparison}

We demonstrate the advantages of the Delaunay flip approach over the degeneration flip approach (\cref{sec:degenflips}) in terms of efficiency as well as numerical robustness. To this end we apply an implementation of the described method and an implementation of the algorithm described by \cite{Campen:2017:SimilarityMaps} (both using standard double precision floating point numbers) to the same set of inputs.

\begin{figure}[b]
    \centering

 \begin{overpic}[scale=1.0]{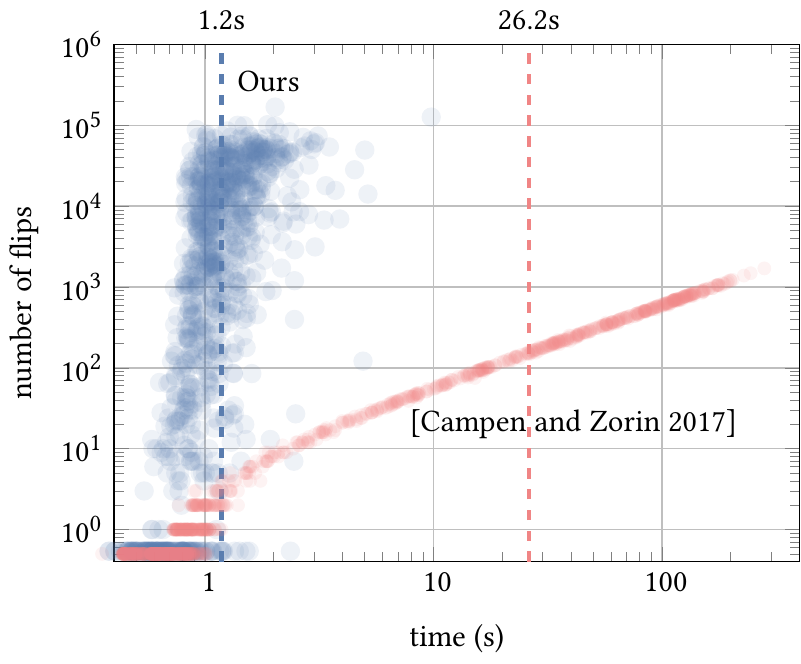}\end{overpic}
 
\vspace{-0.3cm}
    \caption{Scatter plot showing the number of flips and the run time (to reach $\varepsilon_\text{tol} = 10^{-10}$), for the described Delaunay-flip method (blue) and the degeneration flip method (red). Each dot represents one of 1000 test instances. Dashed lines mark the average run time, 1.2s and 26.2s, respectively.
    }
    \label{fig:efficiency}
\end{figure}

\subsubsection*{Efficiency}

The main differences between the two methods lie in the number of linear system solves (to compute the descent direction~$\ve{d}$) and the number of intrinsic flips. In the  proposed method, the number of flips is often significantly higher (see the discussion in \cref{sec:delflips}), while the number of system solves is lower. As a flip is a cheap local operations, while a system solve is an expensive global operation, a run time benefit can be conjectured.

{The scatter plot in \cref{fig:efficiency}} shows that this is the case on average. As test instances we use 1000 different random target angle prescriptions $\ve{\hat\Theta}$ (with $\hat\Theta_i\in (\pi,3\pi)$ for all vertices $v_i$) on a sphere mesh (10K vertices).
Only for relatively simple cases, where the target curvature can be matched without degeneration flips, the number of system solves may be similar such that the non-Delaunay method has a (relatively small) run time benefit due to the lower number of flips. On average, though, run time is $22\times$ lower with the Delaunay-based method on these examples.

\subsubsection*{Robustness}

Differences in robustness can best be observed by considering extreme cases. In \cref{fig:robustcompare} we show the residual error of the two methods when prescribing one very small or very large target angle (while distributing the remaining curvature). For small angles it becomes apparent that the degeneration flip algorithm is numerically more fragile.

\begin{figure}[t]
    \centering
    
\begin{tabular}{@{}cc@{}}
 \begin{overpic}[scale=1.0]{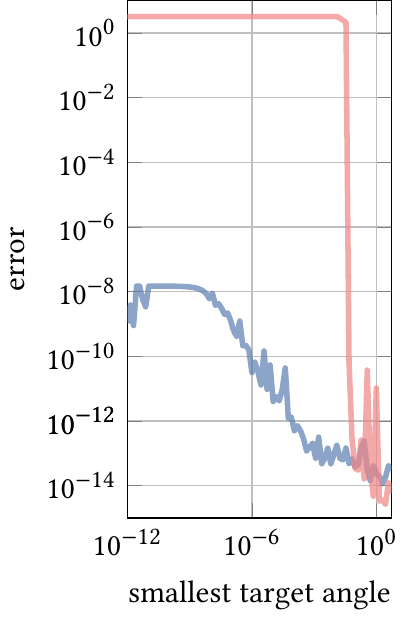}\end{overpic}
&
 \begin{overpic}[scale=1.0]{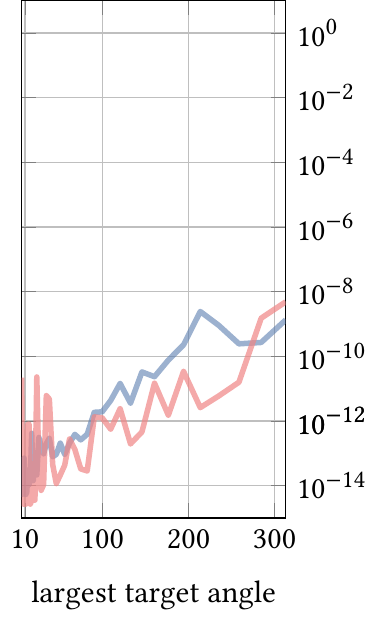}\end{overpic}
\end{tabular}

\vspace{-0.3cm}
    \caption{Final residual angle error $\|\ve{\hat\Theta} - \ve{\Theta}\|_\infty$ for extreme cases (one very small or very large target angle), comparing the Delaunay-based algorithm (blue) and the degeneration flip algorithm \cite{Campen:2017:SimilarityMaps} (red).
    }
    \label{fig:robustcompare}
\end{figure}

\subsection{Accuracy}
\label{sec:res:prec}

\begin{figure}[tb]
    \centering
 \begin{overpic}[scale=1.0]{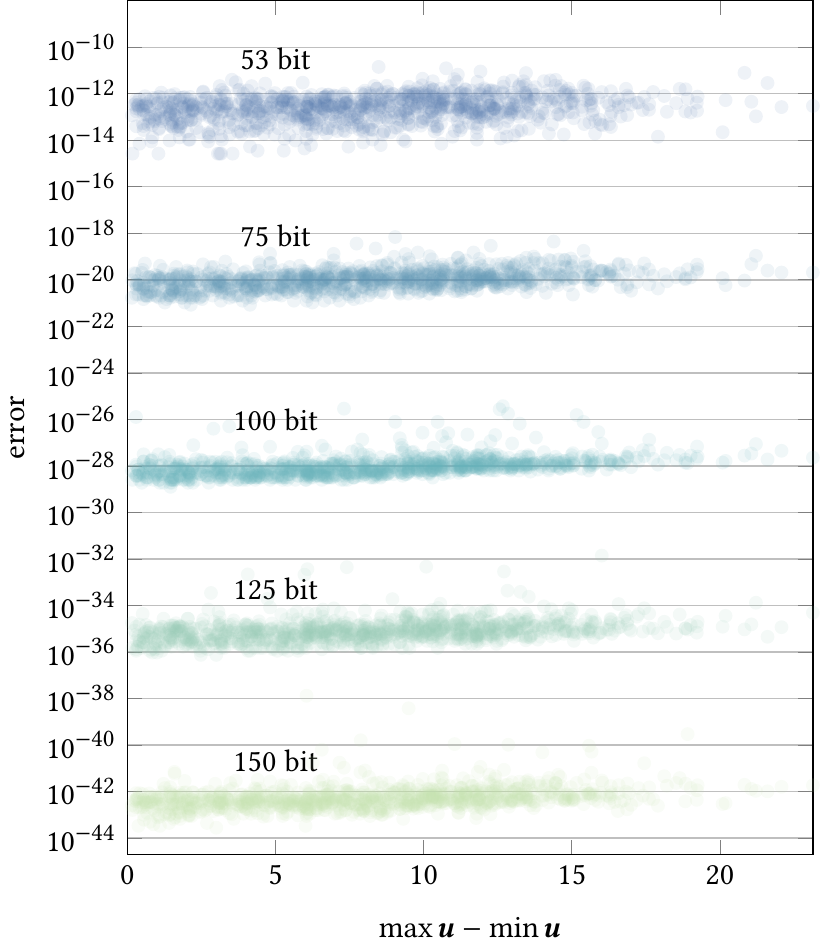}\end{overpic}
\vspace{-0.2cm}
    \caption{Scatter plot showing residual angle error $\|\ve{\hat\Theta} - \ve{\Theta}\|_\infty$ (after at most 50 Newton steps) relative to the range of logarithmic conformal scale factors~$u$. Each dot represents one test instance, run using floating point numbers with a mantissa of 53~bits (\texttt{double}), 75~bits, 100~bits, 125~bits, 150~bits (\texttt{MPFR}).}
    \label{fig:precision}
\end{figure}

While the method is theoretically guaranteed to yield the desired result, in practice numerical inaccuracies limit how closely the target curvature will be matched. As the method involves exponential  and trigonometric functions (\cref{eq:barell,eq:flattening1}), it cannot be implemented in a numerically exact manner using adaptive precision rational or integer number types. Using extended precision floating point number types (such as MPFR), the method's accuracy can, however, be increased arbitrarily. {We evaluate the effect of this choice on result accuracy in \cref{fig:precision}.} As test instances we use 1000 different random target angle prescriptions $\ve{\hat\Theta}$ (with $\hat\Theta_i\in (\pi,3\pi)$ for all vertices $v_i$) on a sphere mesh (1K vertices).

As can be observed, the remaining error that does not vanish due to numerical limitations decreases consistently as the number of bits used for the floating point computations is increased. Due to dependence on many factors (input mesh and edge lengths, target angles, choice of linear system solver for the Newton direction) a simple bound on the error cannot be given, but \cref{fig:precision} gives an empirical idea of the behavior. Note that some correlation can be observed to the conformal scale distortion (the range $[e^{\min\ve{u}},e^{\max\ve{u}}]$) that is required to match the target curvature.

\bibliographystyle{ACM-Reference-Format}
\bibliography{main.bib}

\end{document}